\documentclass[11pt]{article}
\usepackage{hyperref}
\usepackage{times}  
\usepackage{mathpazo}
\usepackage{amssymb,amsmath,amsthm}
\usepackage{epsfig}

\newtheorem{theorem}{Theorem}[section]
\newtheorem{proposition}[theorem]{Proposition}
\newtheorem{definition}[theorem]{Definition}

\newtheorem{claim}[theorem]{Claim}
\newtheorem{lemma}[theorem]{Lemma}

\newtheorem{remark}[theorem]{Remark}

\newcommand{\qedsymb}{\hfill{\rule{2mm}{2mm}}}
\renewenvironment{proof}[1][]{\begin{trivlist}
\item[\hspace{\labelsep}{\bf\noindent Proof#1:\/}] }{\qedsymb\end{trivlist}}

\def\calF{{\cal F}}

\renewcommand{\epsilon}{\varepsilon}

\newcommand{\rank}{\mathop{\mathrm{rank}}}
\newcommand{\minrank}{\mathop{\mathrm{minrk}}}

\newcommand{\tb}{\mathop{\mathrm{tb}}}

\newcommand{\Fset}{\mathbb{F}}         

\begin{document}

\title{{\bf Task-based Solutions to Embedded Index Coding}}

\author{
Ishay Haviv\thanks{School of Computer Science, The Academic College of Tel Aviv-Yaffo, Tel Aviv 61083, Israel.
}
}

\date{}

\maketitle

\begin{abstract}
In the {\em index coding} problem a sender holds a message $x \in \{0,1\}^n$ and wishes to broadcast information to $n$ receivers in a way that enables the $i$th receiver to retrieve the $i$th bit $x_i$. Every receiver has prior side information comprising a subset of the bits of $x$, and the goal is to minimize the length of the information sent via the broadcast channel.
Porter and Wootters have recently introduced the model of {\em embedded index coding}, where the receivers also play the role of the sender and the goal is to minimize the total length of their broadcast information.
An embedded index code is said to be {\em task-based} if every receiver retrieves its bit based only on the information provided by one of the receivers.

This paper studies the effect of the task-based restriction on linear embedded index coding.
It is shown that for certain side information maps there exists a linear embedded index code of length quadratically smaller than that of any task-based embedded index code.
The result attains, up to a multiplicative constant, the largest possible gap between the two quantities. The proof is by an explicit construction and the analysis involves spectral techniques.
\end{abstract}

\section{Introduction}

In the {\em index coding} problem, introduced by Birk and Kol~\cite{BirkKol98}, a sender holds a message $x \in \{0,1\}^n$ and wishes to broadcast information to $n$ receivers $R_1, \ldots, R_n$ in a way that enables each receiver $R_i$ to retrieve its own message $x_i \in \{0,1\}$. For this purpose, the receivers are allowed to use some side information that they have in advance comprising a subset of the bits of $x$. The side information map is naturally represented by a directed graph $G$ on the vertex set $[n] = \{1,2,\ldots,n\}$ that includes a directed edge $(i,j)$ if the receiver $R_i$ knows $x_j$. We will usually consider symmetric side information maps and will thus refer to $G$ as undirected. For a given side information graph $G$, the goal is to design a coding function that maps any $n$-bit message $x \in \{0,1\}^n$ to a broadcast information of as few bits as possible so that the receivers are able to retrieve their messages based on this information and on the side information that they have. For example, for the complete graph on $n$ vertices, which corresponds to the situation where every receiver $R_i$ knows all the bits of $x$ except $x_i$, broadcasting one bit of information that consists of the xor of the $x_i$'s suffices for the receivers to discover their messages. Of special interest is the setting of {\em linear} index coding in which the sender is restricted to apply a linear encoding function over the binary field $\Fset_2$. Bar-Yossef, Birk, Jayram, and Kol~\cite{BBJK06} have shown that the minimum length of a linear index code for a side information graph $G$ is precisely characterized by the minrank parameter denoted ${\minrank}_2(G)$.

In a recent work, Porter and Wootters~\cite{PorterW19} have introduced a variant of the index coding problem called {\em embedded index coding} whose study is motivated by applications in distributed computation.
In this model, the receivers also play the role of the sender, namely, the coding scheme involves a set of receivers each of which broadcasts to all other receivers information that depends only on the messages known to it according to the side information graph. As before, every receiver should be able to retrieve its message based on the broadcasted information and on the side information that it has. The goal here is to minimize the total number of bits broadcasted by all receivers. Note that any coding scheme used for embedded index coding induces a coding scheme of the same length for the sender in the standard centralized setting. On the other hand, it was shown in~\cite{PorterW19} that for every side information graph $G$ (with no isolated vertices) there exists a linear embedded index code whose length is at most twice the length of an optimal linear index code for $G$ in the centralized setting, i.e., at most $2 \cdot {\minrank}_2(G)$.

A special family of solutions to embedded index coding is that of {\em task-based} index codes, defined and studied in~\cite{PorterW19}.
As before, for a given side information graph $G$ a subset of the receivers broadcasts information to all other receivers. However, while in a general embedded index code a receiver is allowed to retrieve its message using the entire broadcast information (and the side information), in a task-based solution every receiver can use only the broadcast information provided by {\em one} of the receivers. In other words, every receiver that plays the role of a sender in a task-based solution is responsible to a subset of its neighborhood in the side information graph, in the sense that the information that it broadcasts enables each receiver in this subset to retrieve its message. Clearly, the length of an optimal task-based embedded index code for a given side information graph $G$ is at least as large as the length of an optimal general embedded index code for $G$.

The study of task-based solutions to embedded index coding is motivated by several aspects of the general embedded index coding problem.
Firstly, task-based solutions seem to be more computationally tractable than those of general embedded index coding.
Indeed, a heuristic algorithm of~\cite{PorterW19} produces a task-based solution to a given instance of embedded index coding by first choosing a partition of the receivers into sets each of which is associated with a sender, and then computing the optimal centralized solutions of the instances induced by each of these sets.
Secondly, task-based solutions are more robust compared to general solutions of embedded index coding, in the sense that a receiver is able to retrieve its message whenever the single sender responsible to it succeeds in broadcasting its information, independently of errors and delays of other senders.
Finally, task-based embedded index coding is related to other notions studied in the area.
This includes instantly decodable network codes~\cite{LeTDM13}, whose study is concerned with maximizing the number of receivers that a sender can handle, and locally decodable index codes~\cite{HavivL12,NatarajanKL18} in which every receiver is allowed to use only a small part of the entire broadcast information.

\subsection{Our Contribution}

The present paper studies the effect of the task-based restriction on linear embedded index coding.
For a side information graph $G$, let $\tb(G)$ denote the minimum total length of a linear task-based embedded index code for $G$.
We first observe that $\tb(G)$ is at most quadratic in the minimum length of a linear index code for $G$ in the centralized setting.

\begin{proposition}\label{prop:upper_t(G)}
For every graph $G$ with no isolated vertices, $\tb(G) \leq O({\minrank}_2(G)^2)$.
\end{proposition}

Our main result is the following matching lower bound.

\begin{theorem}\label{thm:main}
For every integer $k$ there exists a graph $G$ such that $\minrank_2(G)=k$ and $\tb(G)= \Theta(k^2)$.
\end{theorem}
As mentioned before, for a graph $G$ with no isolated vertices the length of an optimal linear embedded index code is at most $2 \cdot {\minrank}_2(G)$ (see Theorem~\ref{thm:minrk_embedded}). Hence, Theorem~\ref{thm:main} provides graphs $G$ for which there exists a linear embedded index code of length quadratically smaller than $\tb(G)$.
This implies an inherent limitation on the algorithm of~\cite{PorterW19} to embedded index coding.
Note that in contrast to the graphs given in Theorem~\ref{thm:main}, for most graphs $G$ the value of $\tb(G)$ is linear in $\minrank_2(G)$.
This follows from the fact that the minrank parameter of a typical random graph is linear in its clique cover number~\cite{Golovnev0W17}.

The proof of Theorem~\ref{thm:main} relies on an explicit graph family defined by Peeters~\cite{Peeters96} (see also~\cite{ChlamtacH14}) and on a spectral technique due to Alon and Krivelevich~\cite{AlonK97}.

\subsection{Related Work}

The index coding problem, introduced in~\cite{BirkKol98} and further developed in~\cite{BBJK06}, has been studied in various variations and extensions.
This research is motivated by applications such as distributed storage~\cite{mazumdar2014duality}, wireless communication~\cite{jafar2014topological}, and the more general problem of network coding~\cite{AhlswedeCLY00}.
The variant called embedded index coding, introduced in~\cite{PorterW19}, can be viewed as a special case of the multi-sender index coding model studied in~\cite{OngHL16} which allows multiple senders and multiple receivers but as disjoint sets of vertices (see also~\cite{LiOJ19}).
The framework of index coding studied in~\cite{PorterW19} is more general than the one considered in the current work and allows the receivers to request multiple messages.

A significant attention was given in the literature to the study of linear index coding which is characterized, as shown in~\cite{BBJK06}, by the minrank parameter (see Definition~\ref{def:minrank}). This graph parameter has been originally defined in 1979 by Haemers~\cite{Haemers79} in the study of the Shannon capacity of graphs and has later found a useful equivalent definition based on a graph family introduced by Peeters in~\cite{Peeters96} (see Section~\ref{sec:proofs_2}).
This graph family was used in~\cite{Peeters96} to obtain relations between the minrank of a graph and the chromatic number of its complement, and it was further investigated in~\cite{ChlamtacH14} where its spectral properties were involved in the analysis of an approximation algorithm for minrank based on semi-definite programming (see also~\cite{HavivTheta18}).
Our proof of Theorem~\ref{thm:main} relies on the graph family from~\cite{Peeters96} and combines its spectral properties proved in~\cite{ChlamtacH14} with a result of~\cite{AlonK97} on pseudo-random graphs.
Our approach is inspired by a work of Vinh~\cite{Vinh08a} who studied the number of orthogonal vector sets in large subsets of vector spaces over finite fields.

\subsection{Outline}

The rest of the paper is organized as follows.
In Section~\ref{sec:preliminaries}, we gather several definitions and results needed throughout the paper.
In Section~\ref{sec:proofs}, we prove Proposition~\ref{prop:upper_t(G)} and Theorem~\ref{thm:main} and provide an analogue of Theorem~\ref{thm:main} for non-linear index coding.
We end the paper in Section~\ref{sec:remarks} with a few concluding remarks.

\section{Preliminaries}\label{sec:preliminaries}

For a graph $G=(V,E)$, we let $N(i)$ denote the set of vertices in $V$ adjacent to a vertex $i \in V$. We also let $G[U]$ denote the subgraph of $G$ induced by a subset $U$ of $V$.
For an $n$-dimensional vector $x$ and a set $A \subseteq [n]$, we let $x|_A$ denote the restriction of $x$ to the indices in $A$.

\paragraph{Index coding.}
We turn to formally define the variants of the index coding problem considered in this work.
Since the graphs in our construction in Theorem~\ref{thm:main} are undirected, we restrict our attention to the undirected case.

\begin{definition}[Index Coding]
Let $G$ be a side information graph on the vertex set $[n]$.
\begin{enumerate}
  \item\label{itm:1} A {\em linear index code} of length $\ell$ for $G$ is a linear encoding function $E: \Fset_2^n \rightarrow \Fset_2^\ell$ for which there exist $n$ linear decoding functions $D^{(i)}: \Fset_2^{\ell + |N(i)|} \rightarrow \Fset_2$ ($i \in [n]$) such that the following holds: For all $x \in \Fset_2^n$ and $i \in [n]$, $D^{(i)}(E(x),x|_{N(i)}) = x_i$.
  \item\label{itm:2} A {\em linear embedded index code} of length $\ell$ for $G$ is a collection of linear encoding functions $E^{(j)}: \Fset_2^{|N(j)|} \rightarrow \Fset_2^{\ell_j}$ ($j \in S$ for some $S \subseteq [n]$) where $\ell = \sum_{j \in S}{\ell_j}$, for which there exist $n$ linear decoding functions $D^{(i)}: \Fset_2^{\ell + |N(i)|} \rightarrow \Fset_2$ ($i \in [n]$) such that the following holds: For all $x \in \Fset_2^n$ and $i \in [n]$, $D^{(i)}( (E^{(j)}(x|_{N(j)}))_{j \in S},x|_{N(i)}) = x_i$.
  \item\label{itm:3} A {\em linear task-based embedded index code} of length $\ell$ for $G$ is a collection of linear encoding functions $E^{(j)}: \Fset_2^{|N(j)|} \rightarrow \Fset_2^{\ell_j}$ ($j \in S$ for some $S \subseteq [n]$) where $\ell = \sum_{j \in S}{\ell_j}$, for which there exist indices $j_1,\ldots,j_n \in S$ and $n$ linear decoding functions $D^{(i)}: \Fset_2^{\ell_{j_i} + |N(i)|} \rightarrow \Fset_2$ ($i \in [n]$) such that the following holds: For all $x \in \Fset_2^n$ and $i \in [n]$, $D^{(i)}( E^{(j_i)}(x|_{N(j_i)}),x|_{N(i)}) = x_i$.
\end{enumerate}
\end{definition}

\begin{remark}
Note that a graph with isolated vertices cannot have an embedded index code.
\end{remark}

The minrank parameter over $\Fset_2$ is defined as follows.

\begin{definition}\label{def:minrank}
Let $G=([n],E)$ be a directed graph.
We say that an $n$ by $n$ matrix $M$ over $\Fset_2$ {\em represents} $G$ if $M_{i,i} \neq 0$ for every $i \in [n]$, and $M_{i,j}=0$ for every distinct $i,j \in [n]$ such that $(i,j) \notin E$.
The {\em minrank} of $G$ over $\Fset_2$ is defined as
\[{\minrank}_2(G) =  \min\{{\rank}_{\Fset_2}(M)\mid M \mbox{ represents }G \}.\]
The definition is naturally extended to undirected graphs by replacing every undirected edge with two oppositely directed edges.
\end{definition}
\noindent
Notice that every (undirected) graph $G$ satisfies ${\minrank}_2(G) \geq \alpha(G)$, where $\alpha(G)$ stands for the independence number of $G$.

The minimum length of a linear index code (Definition~\ref{def:minrank}, Item~\ref{itm:1}) was characterized by the minrank parameter in~\cite{BBJK06}.
\begin{theorem}[\cite{BBJK06}]
For every graph $G$, the minimum length of a linear index code for $G$ is ${\minrank}_2(G)$.
\end{theorem}

The minimum length of a linear embedded index code (Definition~\ref{def:minrank}, Item~\ref{itm:2}) was bounded in~\cite{PorterW19} from above using the minrank parameter.

\begin{theorem}[\cite{PorterW19}]\label{thm:minrk_embedded}
For every graph $G$ with no isolated vertices, the minimum length of a linear embedded index code for $G$ is at most $2 \cdot {\minrank}_2(G)$.
\end{theorem}

A {\em neighborhood partition} of a graph $G=(V,E)$ is a partition $(N_i)_{i \in S}$ of the vertex set $V$, where $S \subseteq V$ and $\emptyset \neq N_i \subseteq N(i)$ for all $i \in S$.
For a graph $G$, let $\tb(G)$ denote the minimum length of a linear task-based embedded index code  for $G$ (Definition~\ref{def:minrank}, Item~\ref{itm:3}). This quantity was characterized in~\cite{PorterW19} using the minrank parameter and the notion of neighborhood partitions.

\begin{lemma}[\cite{PorterW19}]\label{lemma:t(G)_eq_def}
For every graph $G$ with no isolated vertices, $\tb(G)$ is the minimum of
\[\sum_{i \in S}{{\minrank}_2(G[N_i])}\]
over all possible neighborhood partitions $(N_i)_{i \in S}$ of $G$.
\end{lemma}

\section{Proofs}\label{sec:proofs}

\subsection{Proof of Proposition~\ref{prop:upper_t(G)}}

A {\em dominating set} in a graph $G=(V,E)$ is a subset $D \subseteq V$ of the vertex set such that every vertex of $G$ either belongs to $D$ or is adjacent to a vertex of $D$. Let $\gamma(G)$ denote the minimum size of a dominating set in a graph $G$.
We prove the following bound.

\begin{proposition}\label{prop:upper_t(G)_full}
For every graph $G$ with no isolated vertices, $\tb(G) \leq \gamma(G) \cdot ({\minrank}_2(G)+1)$.
\end{proposition}

\begin{proof}
Let $G=(V,E)$ be a graph.
Consider a dominating set $D \subseteq V$ in $G$ of minimum size and denote its vertices by $D = \{i_1, \ldots, i_{d}\}$ where $d= \gamma(G)$.
We define a neighborhood partition of $G$ as follows. First, for every $j \in [d]$ consider the set $N_{i_j} = N(i_j) \setminus \cup_{h<j}{N(i_h)}$, that is, the set of vertices that are adjacent to $i_j$ but not to any $i_h$ with $h<j$. Next, since $D$ is a dominating set, the only vertices that have not been covered yet belong to $D$, and since $G$ has no isolated vertices they can be partitioned into sets $(N_i)_{i \in S}$ for some $S \subseteq V \setminus D$ where $\sum_{i \in S}{|N_i|} \leq d$ and $N_i \subseteq N(i)$ for all $i \in S$.
It follows that the sets of $(N_{i_j})_{j \in [d]}$ and $(N_i)_{i \in S}$ (omitting empty sets, if any) form a neighborhood partition of $G$.
Moreover, we clearly have ${\minrank}_2(G[N_{i_j}]) \leq {\minrank}_2(G)$ for every $j \in [d]$ and ${\minrank}_2(G[N_{i}]) \leq |N_i|$ for every $i \in S$. By Lemma~\ref{lemma:t(G)_eq_def}, it follows that
\begin{eqnarray*}
\tb(G) &\leq& \sum_{j \in [d]}{{\minrank}_2(G[N_{i_j}])} + \sum_{i \in S}{{\minrank}_2(G[N_i])} \\
&\leq& d \cdot {\minrank}_2(G)+\sum_{i \in S}{|N_i|} \leq d \cdot ({\minrank}_2(G)+1),
\end{eqnarray*}
and we are done.
\end{proof}

The assertion of Proposition~\ref{prop:upper_t(G)} follows now easily.
\begin{proof}[ of Proposition~\ref{prop:upper_t(G)}]
Observe that every graph $G$ satisfies $\gamma(G) \leq \alpha(G) \leq {\minrank}_2(G)$.
By Proposition~\ref{prop:upper_t(G)_full}, it follows that a graph $G$ with no isolated vertices satisfies $\tb(G) \leq O({\minrank}_2(G)^2)$, as desired.
\end{proof}

\subsection{Proof of Theorem~\ref{thm:main}}\label{sec:proofs_2}

The proof of Theorem~\ref{thm:main} relies on a graph family introduced in~\cite{Peeters96}, defined as follows.

\paragraph{The Graph Family $G_k$.}

For an integer $k \geq 1$ we define the (undirected) graph $G_k=(V,E)$ on the vertex set
\[V = \{(u,v) \in \Fset_2^k \times \Fset_2^k\mid \langle u,v \rangle = 1 \},\]
in which two vertices $(u_1,v_1)$ and $(u_2,v_2)$ are adjacent if and only if $\langle u_1,v_2 \rangle = \langle v_1,u_2 \rangle = 0$.
Observe that $|V| = (2^k-1)\cdot 2^{k-1}$ and that $G_k$ is regular with degree $(2^{k-1}-1) \cdot 2^{k-2}$.
It is easy to show that the minrank over $\Fset_2$ of the complement $\overline{G_k}$ of the graph $G_k$ is precisely $k$.

\begin{claim}\label{claim:minrank(G_k_comp)}
For every $k \geq 1$, $\alpha(\overline{G_k}) = {\minrank}_2(\overline{G_k}) = k$.
\end{claim}

\begin{proof}
By the definition of $G_k=(V,E)$, every vertex $x \in V$ is associated with a pair $(u_x,v_x) \in \Fset_2^k \times \Fset_2^k$ satisfying $\langle u_x,v_x \rangle =1$.
Let $M_1$ and $M_2$ be the $k \times |V|$ matrices over $\Fset_2$ with columns indexed by $V$, such that the column associated with vertex $x$ in $M_1$ consists of the vector $u_x$ and the column associated with it in $M_2$ consists of the vector $v_x$.
The matrix $M = M_1^T \cdot M_2$ represents the graph $\overline{G_k}$, because for every $x \in V$ we have $\langle u_x,v_x \rangle =1$ whereas every distinct vertices $(u_x,v_x)$ and $(u_y,v_y)$ that are not adjacent in $\overline{G_k}$ (i.e., adjacent in $G_k$) satisfy $\langle u_x,v_y \rangle = \langle u_y,v_x \rangle = 0$. Since $M$ has rank at most $k$ over $\Fset_2$, it follows that ${\minrank}_2(\overline{G_k}) \leq k$.
On the other hand, the set of vertices $\{(e_i,e_i)\}_{i \in [k]}$, where $e_i$ denotes the vector in $\Fset_2^k$ that has a nonzero entry only in the $i$th coordinate, forms an independent set in $\overline{G_k}$. Since the size of an independent set in a graph bounds from below its minrank, we obtain $k \leq \alpha(\overline{G_k})  \leq {\minrank}_2(\overline{G_k}) \leq k$, and we are done.
\end{proof}

The graph family $G_k$ can be used to provide an alternative definition for the minrank parameter over $\Fset_2$. Indeed, it is straightforward to verify that for every graph $G$, ${\minrank}_2(G)$ is the smallest integer $k$ for which there exists a homomorphism from $\overline{G}$ to $G_k$.
This means, in a sense, that the graph $\overline{G_k}$ captures the structure of all graphs with minrank $k$, and as such, it is natural to consider it for obtaining a graph $G$ with minrank $k$ and yet a significantly larger $\tb(G)$.
This is precisely the approach taken in our proof of Theorem~\ref{thm:main}.
To prove the lower bound on $\tb(\overline{G_k})$ we show, roughly speaking, that every economical neighborhood partition of $\overline{G_k}$ includes $\Omega(k)$ vertices $i$ associated with a `large' neighborhood $N_i$. For those vertices $i$, it is shown that the subgraph of $\overline{G_k}$ induced by $N_i$ contains an independent set of size linear in $k$, implying that its minrank is linear in $k$ as well. This yields, using Lemma~\ref{lemma:t(G)_eq_def}, that the length of any linear task-based embedded index code for $\overline{G_k}$ is at least of order $k^2$. The existence of large independent sets in the induced subgraphs of $\overline{G_k}$ is proved by a spectral technique, described next.

The graph $G_k$ was shown in~\cite{ChlamtacH14} to be vertex-transitive and edge-transitive. Moreover, the strong symmetry properties of $G_k$ were used there to exactly determine its eigenvalues (i.e., the eigenvalues of its adjacency matrix).

\begin{lemma}[\cite{ChlamtacH14}]\label{lemma:lambda_n}
For every $k \geq 3$, the second largest eigenvalue of $G_k$ in absolute value is $2^{3k/2-3}$.
\end{lemma}

An {\em $(n,d,\lambda)$-graph} is a $d$-regular graph on $n$ vertices in which all eigenvalues, but the largest one, are of absolute value at most $\lambda$.
It is well known that $(n,d,\lambda)$-graphs with $\lambda$ much smaller than $d$ have various pseudo-random properties (see, e.g.,~\cite{KrivelevichS2006}).
In particular, the following result from~\cite{AlonK97} says that every sufficiently large subset of the vertex set of an $(n,d,\lambda)$-graph contains a large complete graph.

\begin{proposition}[\cite{AlonK97}]\label{prop:(n,d,l)-exists}
Let $G$ be an $(n,d,\lambda)$-graph. Then for every integer $r \geq 2$ and every subset $U$ of the vertex set of $G$ satisfying
\[ |U| > \frac{(\lambda+1)n}{d} \cdot \Big ( 1+\frac{n}{d}+\cdots+\Big (\frac{n}{d} \Big )^{r-2}\Big ),\]
the graph $G[U]$ contains a copy of the complete graph $K_r$.
\end{proposition}

Applying Proposition~\ref{prop:(n,d,l)-exists} to the graph $G_k$, we obtain the following result.

\begin{lemma}\label{lemma:K_r-in-G_k}
There exists a constant $c>0$ such that for all integers $k \geq 3$ and $r \geq 2$ and for every subset $U$ of the vertex set of $G_k$ satisfying $|U| \geq c \cdot 2^{3k/2+2r}$, the graph $G_k[U]$ contains a copy of $K_r$.
\end{lemma}

\begin{proof}
Let $k \geq 3$ and $r \geq 2$ be integers.
By Lemma~\ref{lemma:lambda_n}, the graph $G_k$ is an $(n,d,\lambda)$-graph for
\[n = (2^k-1) \cdot 2^{k-1},~~~ d= (2^{k-1}-1) \cdot 2^{k-2},~~~ \mbox{and}~~~\lambda = 2^{3k/2-3}.\]
Observe that $\frac{n}{d} = 4 \cdot (1+\frac{1}{2^k-2})$ and that $\lambda \cdot (\frac{n}{d})^{r-1} = \Theta( 2^{3k/2+2r} )$, where we have used the assumption that, say, $r \leq k/4$ (Otherwise the assertion of the lemma trivially holds, because there is no subset $U$ of the vertex set of $G_k$ with the required size).
By Proposition~\ref{prop:(n,d,l)-exists}, for every subset $U$ of the vertex set of $G_k$ satisfying $|U| \geq \Omega( \lambda \cdot (\frac{n}{d})^{r-1}) = \Omega( 2^{3k/2+2r} )$, the graph $G_k[U]$ contains a copy of $K_r$, so we are done.
\end{proof}

\begin{remark}
An equivalent statement to that of Lemma~\ref{lemma:K_r-in-G_k} is the following.
For integers $k \geq 3$ and $r \geq 2$, let $\calF \subseteq \Fset_2^k \times \Fset_2^k$ be a collection of non-orthogonal vector pairs such that $|\calF| \geq c \cdot 2^{3k/2+2r}$ where $c>0$ is an absolute constant. Then there exist $r$ pairs $(u_1,v_1),\ldots,(u_r,v_r) \in \calF$ such that $u_i$ and $v_j$ are orthogonal whenever $i \neq j$.
\end{remark}

We need the following simple linear algebra lemma.

\begin{lemma}\label{lemma:subspaces}
Let $W_1, W_2 \subseteq \Fset_2^k$ be two subspaces of dimension at least $k-\ell$.
Then the number of pairs $(w_1,w_2) \in W_1 \times W_2$ satisfying $\langle w_1, w_2 \rangle = 1$ is at least $(2^{k-\ell}-2^\ell) \cdot 2^{k-\ell-1}$.
\end{lemma}

\begin{proof}
For every vector $w_1 \in W_1$ the number of vectors $w_2 \in W_2$ satisfying $\langle w_1, w_2 \rangle = 1$ depends on whether $w_1$ belongs to the orthogonal complement of $W_2$ or not: If $w_1 \in W_2^{\perp}$ then there are no such vectors $w_2$ and otherwise their number is $|W_2|/2$. By $\dim(W_2) \geq k-\ell$, we have $|W_2^{\perp}| \leq 2^\ell$. This implies, using $\dim(W_1) \geq k-\ell$, that for at least $2^{k-\ell}-2^{\ell}$ of the vectors $w_1 \in W_1$ there are $|W_2|/2 \geq 2^{k-\ell-1}$ vectors $w_2 \in W_2$ satisfying $\langle w_1, w_2 \rangle = 1$. Hence, the total number of the required pairs is  at least $(2^{k-\ell}-2^\ell) \cdot 2^{k-\ell-1}$.
\end{proof}

Equipped with Lemmas~\ref{lemma:K_r-in-G_k} and~\ref{lemma:subspaces}, we are ready to prove Theorem~\ref{thm:main}.

\begin{proof}[ of Theorem~\ref{thm:main}]
For a given integer $k$, we prove that the complement $\overline{G_k}$ of the graph $G_k$ satisfies the assertion of the theorem. By Claim~\ref{claim:minrank(G_k_comp)} we have $\minrank_2(\overline{G_k})=k$, and by Proposition~\ref{prop:upper_t(G)} we have $\tb(\overline{G_k}) \leq O(k^2)$. We turn to prove the lower bound on $\tb(\overline{G_k})$.
Because of the asymptotic nature of the bound, we may assume from now on that $k$ is sufficiently large.

By Lemma~\ref{lemma:t(G)_eq_def}, it suffices to show that every neighborhood partition $(N_i)_{i \in S}$ of $\overline{G_k}$ satisfies
\[\sum_{i \in S}{{\minrank}_2(\overline{G_k}[N_i])} = \Omega(k^2).\]
Let $(N_i)_{i \in S}$ be a neighborhood partition of $\overline{G_k}$. It can be assumed that $|S| \leq k^2$ as otherwise there is nothing to prove. Denote $r = \lfloor k/8 \rfloor$, and let $L \subseteq S$ be the collection of vertices $i \in S$ such that $|N_i| \geq c \cdot 2^{3k/2+2r}$, where $c$ is the positive constant from Lemma~\ref{lemma:K_r-in-G_k}. Note that
\begin{eqnarray}\label{eq:S-L}
\sum_{i \in S \setminus L}{|N_i|} \leq c \cdot 2^{3k/2+2r} \cdot k^2 \leq c \cdot k^2 \cdot 2^{7k/4}.
\end{eqnarray}

Denote $L = \{ (u_1,v_1),\ldots,(u_\ell,v_\ell) \}$ where $\ell = |L|$.
We turn to prove that $\ell$ must be linear in $k$.
To this end, consider the subspaces of $\Fset_2^k$ defined by
$W_1 = \{ w \in  \Fset_2^k \mid \langle w,v_j \rangle =0 \mbox{ for all }j \in [\ell] \}$ and
$W_2 = \{ w \in  \Fset_2^k \mid \langle w,u_j \rangle =0 \mbox{ for all }j \in [\ell] \}$,
and notice that each of them has dimension at least $k-\ell$.
By the definition of $G_k$, the pairs $(w_1,w_2) \in W_1 \times W_2$ such that $\langle w_1, w_2 \rangle =1$ are vertices that are not adjacent in $\overline{G_k}$ to any of the vertices of $L$, hence they must be covered by the sets of $(N_i)_{i \in S \setminus L}$.
By Lemma~\ref{lemma:subspaces}, the number of these vertices is at least $(2^{k-\ell}-2^\ell) \cdot 2^{k-\ell-1}$, so using~\eqref{eq:S-L} we obtain that $(2^{k-\ell}-2^\ell) \cdot 2^{k-\ell-1} \leq c \cdot k^2 \cdot 2^{7k/4}$.
This inequality, unless $\ell \geq k/2$, implies that $2^{2(k-\ell-1)} \leq c \cdot k^2 \cdot 2^{7k/4}$, and thus $\ell \geq (1/8-o(1)) \cdot k$, as desired.

Finally, for every $i \in L$ we have $|N_i| \geq c \cdot 2^{3k/2+2r}$, hence by Lemma~\ref{lemma:K_r-in-G_k} the subgraph of $G_k$ induced by $N_i$ contains a complete graph $K_r$, that is, $\alpha(\overline{G_k}[N_i]) \geq r$ for every $i \in L$. We derive that
\[\sum_{i \in S}{{\minrank}_2(\overline{G_k}[N_i])} \geq \sum_{i \in L}{\alpha(\overline{G_k}[N_i])} \geq \ell \cdot r \geq \Omega(k^2),\]
and we are done.
\end{proof}

\subsection{Task-based Non-linear Index Coding}

For a graph $G$ on the vertex set $[n]$, let $\beta_1(G)$ denote the minimum length of a general (i.e., not necessarily linear) index code for $G$ over the binary alphabet. Namely, $\beta_1(G)$ is the smallest integer $\ell$ such that there exist an encoding function $E: \{0,1\}^n \rightarrow \{0,1\}^\ell$ and $n$ decoding functions $D^{(i)}: \{0,1\}^{\ell + |N(i)|} \rightarrow \{0,1\}$ ($i \in [n]$) such that the following holds: For all $x \in \{0,1\}^n$ and $i \in [n]$, $D^{(i)}(E(x),x|_{N(i)}) = x_i$. It is well known and easy to check that every graph $G$ satisfies
\begin{eqnarray}\label{eq:beta_1}
\alpha(G) \leq \beta_1(G) \leq {\minrank}_2(G).
\end{eqnarray}
While the model of embedded index coding was defined and studied in~\cite{PorterW19} for the linear setting, it is natural to consider it for the more general setting where the encoding and decoding functions are not necessarily linear.
The linearity restriction on index coding may sometimes significantly affect the length of the broadcast information (see, e.g.,~\cite{LS07}). However, we observe here that the proof of Theorem~\ref{thm:main} can be used to provide a quadratic gap between the minimum length of an index code and the minimum length of a task-based index code for the non-linear setting as well.

\begin{proposition}
For every integer $k$ there exists a graph $G$ with $\beta_1(G)=k$ such that the minimum total length of a (not necessarily linear) task-based embedded index code for $G$ is $\Theta(k^2)$.
\end{proposition}

\begin{proof}[ (sketch)]
For a given integer $k$, we claim that the complement $\overline{G_k}$ of the graph $G_k$ satisfies the assertion of the proposition.
First, by Claim~\ref{claim:minrank(G_k_comp)} combined with~\eqref{eq:beta_1} we obtain that $\beta_1(\overline{G_k}) = k$.
Next, the minimum total length of a task-based index code for $\overline{G_k}$ is clearly bounded from above by $\tb(\overline{G_k})$ which was shown in Theorem~\ref{thm:main} to be $\Theta(k^2)$.
For the lower bound, observe that the minimum total length of a task-based index code for any graph $G$ is the minimum of
\begin{eqnarray}\label{eq:non-linear}
\sum_{i \in S}{\beta_1(G[N_i])}
\end{eqnarray}
over all possible neighborhood partitions $(N_i)_{i \in S}$ of $G$. This follows from the fact that every task-based index code for $G$ naturally induces a neighborhood partition with a set of receivers for every receiver that plays the role of a sender.
Recalling that the proof of Theorem~\ref{thm:main} provides an $\Omega(k^2)$ lower bound on the minimum of $\sum_{i \in S}{\alpha(\overline{G_k}[N_i])}$ over all neighborhood partitions $(N_i)_{i \in S}$ of $\overline{G_k}$, applying again~\eqref{eq:beta_1} we get the same lower bound for the quantity in~\eqref{eq:non-linear}, and we are done.
\end{proof}

\section{Concluding Remarks}\label{sec:remarks}

\begin{itemize}
  \item Lemma~\ref{lemma:K_r-in-G_k} shows that every induced subgraph of $G_k$ on $\Omega( 2^{3k/2+2r} )$ vertices must include a copy of the complete graph $K_r$. The proof is based on spectral properties of $G_k$ and on a pseudo-random property of $(n,d,\lambda)$-graphs that guarantees the existence of a large complete graph $K_r$ in any sufficiently large induced subgraph (see Proposition~\ref{prop:(n,d,l)-exists}). In fact, such subgraphs are known even to include many copies of $K_r$, just as expected in a random graph with edge probability $d/n$ (see~\cite[Theorem~4.10]{KrivelevichS2006}).
     For the graph family $G_k$, one can show that for every $r \geq 2$ and for every strict subset $U$ of the vertex set of $G_k$ satisfying $|U| =m \geq \omega( 2^{3k/2+2r})$, the graph $G_k[U]$ contains $(1+o(1)) \cdot \frac{m^r}{r!} \cdot 4^{-\binom{r}{2}}$ copies of $K_r$.

  \item The bound provided by Lemma~\ref{lemma:K_r-in-G_k} on the number of vertices in an induced subgraph of $G_k$ that guarantees the existence of $K_r$ suffices for us to obtain a tight lower bound, up to a multiplicative constant, for the question studied in this work (see Theorem~\ref{thm:main}). Nevertheless, it will be interesting to better understand the minimum number of vertices needed for Lemma~\ref{lemma:K_r-in-G_k} to hold. It seems plausible that the bound given there, whose proof relies on spectral techniques, can be somewhat improved. For example, for the special case of $r=2$ the spectral analysis implies that any independent set in $G_k$ has size at most $O(2^{3k/2})$, whereas Alon~\cite{ProofByAuthority} has proved an improved upper bound of $2^{(1+o(1)) \cdot k}$ using different techniques (see~\cite[Section~5]{ChlamtacH14}).
\end{itemize}

\section*{Acknowledgement}
We are grateful to the anonymous reviewers for their valuable suggestions.

\bibliographystyle{abbrv}
\bibliography{task-based-IC}

\end{document}